\documentclass{article}
\usepackage[utf8]{inputenc}

\usepackage{amsmath,amsfonts,latexsym,graphicx,amsthm}
\usepackage{amssymb}
\usepackage{color}
\usepackage[linesnumbered,boxed,ruled,vlined]{algorithm2e}
\usepackage[breaklinks=true]{hyperref}
\usepackage{extarrows}

\hypersetup{colorlinks,linkcolor=blue,citecolor=blue,urlcolor=blue}
\newtheorem{Def}{Definition}
\newtheorem{Thm}{Theorem}
\newtheorem{Prop}{Proposition}
\newtheorem{Exam}{Example}

\newtheorem{Lem}{Lemma}

\def\<{\langle}
\def\>{\rangle}

\newcommand{\cA}{\mathcal{A}}

\newcommand{\cH}{\mathcal{H}}
\newcommand{\cX}{\mathcal{X}}

\newcommand{\tr}{\mathrm{tr}}
\newcommand{\ip}[2]{\langle #1 , #2\rangle}
\newcommand{\pos}[1]{\mathrm{Pos}\left(#1\right)}

\title{Quantum Coupling and Strassen Theorem}
\author{
  Li Zhou\footnote{
    Department of Computer Science and Technology, Tsinghua University, Beijing, China; and Centre for Quantum Software and Information, School of Software, Faculty of Engineering and Information
Technology, University of Technology Sydney, NSW, Australia}
  \and
    Shenggang Ying\footnote{
  Centre for Quantum Software and Information, School of Software, Faculty of Engineering and Information
Technology, University of Technology Sydney, NSW, Australia
  }
\and
 Nengkun Yu\footnote{%
    Centre for Quantum Software and Information, School of Software, Faculty of Engineering and Information
Technology, University of Technology Sydney, NSW, Australia}
 \and
 Mingsheng Ying\footnote{Centre for Quantum Software and Information, School of Software, Faculty of Engineering and Information
Technology, University of Technology Sydney, NSW, Australia; and Institute of Software, Chinese Academy of Sciences, Beijing 100190, China; and Department of Computer Science and Technology, Tsinghua University, Beijing 100084, China
  }
}

\date{\today}

\begin{document}

\maketitle

\begin{abstract}
We introduce a quantum generalisation of the notion of coupling in probability theory. Several interesting examples and basic properties of quantum couplings are presented. In particular, we prove a quantum extension of Strassen theorem for probabilistic couplings, a fundamental theorem in probability theory that can be used to bound the probability of an event in a distribution by the probability of an event in another distribution coupled with the first.
\end{abstract}

\section{Introduction} Coupling is a powerful technique in probability theory, with which random variables can be linked to or compared with each other. It has been widely used in the studies of random walks and Markov chains, interacting particle systems and diffusions, just name a few, in order to establish limit theorems about them, to develop approximations for them, or to derive correlation inequalities between them \cite{Lin02}.

Recently, a very successful application of coupling in computer science was discovered by Barthe et al. \cite{Bar09} that it can serve as a solid mathematical foundation for defining the semantics of probabilistic relational Hoare logic. This discovery enables them to develop a series of powerful proof techniques for reasoning about relational properties of probabilistic computations, in particular, for verification of cryptographic protocols and differential privacy \cite{Bar16, Bar16a, Bar17, Hsu17}.

There is a simple and natural correspondence between probability theory and quantum theory: probability distributions/density operators (mixed quantum states), marginal distributions/partial traces, and more. This correspondence suggests us to explore the possibility of generalising the coupling techniques for reasoning about quantum systems. We expect that these techniques can help us to extend quantum Hoare logic \cite{Ying11} for proving relational properties between quantum programs and further for verifying quantum cryptographic protocols and differential privacy in quantum computation \cite{Zhou17}. But in this paper, we focus on studying quantum couplings themselves.

Strassen theorem \cite{Str65} is a fundamental theorem in probability theory that can be used to bound the probability of an event in a distribution by the probability of an event in another distribution coupled with the first.
The main technical contribution of this paper is proving an elegant (in our opinion) quantum generalisation of Strassen theorem.

\section{Background and Basic Definitions}

\subsection{Probabilistic Coupling}
For convenience of the reader, we first briefly recall the basics of probabilistic coupling, following \cite{Hsu17}. Let $\mathcal{A}$ be a finite or countably infinite set. A sub-distribution over $\mathcal{A}$ is a mapping $\mu:\mathcal{A}\rightarrow [0,1]$ such that $\sum_{a\in\mathcal{A}}\mu (a)\leq 1$. In paricular, if $\sum_{a\in\mathcal{A}}\mu (a)= 1$, then $\mu$ is called a distribution over $\mathcal{A}$. For a sub-distribution $\mu$ over $\mathcal{A}$, we define: \begin{enumerate}\item The weight of $\mu$ is $|\mu|=\sum_{a\in\mathcal{A}}\mu(a);$
\item The support of $\mu$ is ${\rm supp}(\mu)=\{a\in\mathcal{A}: \mu(a)>0\};$ \item The probability of an event $S\subseteq\mathcal{A}$ is $\mu(S)=\sum_{a\in S}\mu(a).$\end{enumerate} Moreover, let $\mu$ be a joint sub-distribution, i.e. a sub-distribution over Cartesian product $\mathcal{A}_1\times\mathcal{A}_2$. Then its marginals $\pi_1(\mu), \pi_2(\mu)$ over $\mathcal{A}_1$ and $\mathcal{A}_2$ are, respectively, defined by \begin{align*}\pi_1(\mu)(a_1)&=\sum_{a_2\in\mathcal{A}_2}\mu(a_1,a_2)\ {\rm for\ every}\ a_1\in\mathcal{A}_1,\\
\pi_2(\mu)(a_2) &=\sum_{a_1\in\mathcal{A}_1}\mu(a_1,a_2)\ {\rm for\ every}\ a_2\in\mathcal{A}_2.\end{align*}

Now we can define the notion of coupling.

\begin{Def}[Probabilistic Coupling]Let $\mu_1, \mu_2$ be sub-distributions over $\mathcal{A}_1,\mathcal{A}_2$, respectively. Then a sub-distribution $\mu$ over $\mathcal{A}_1\times\mathcal{A}_2$ is called a coupling for $(\mu_1,\mu_2)$ if $\pi_1(\mu)=\mu_1$ and $\pi_2(\mu)=\mu_2$.
\end{Def}

Here are some simple examples of coupling taken from \cite{Hsu17}.

\begin{Exam}\label{exam-bit} Let $\mathbf{Flip}$ be the uniform distribution over booleans, i.e. $\mathbf{Flip}(0)=\mathbf{Flip}(1)=\frac{1}{2}$. Then the following are two couplings for $(\mathbf{Flip},\mathbf{Flip})$:\begin{enumerate}
\item Identity coupling: $\mu_{\rm id}(a_1,a_2)=\begin{cases}\frac{1}{2}\ &{\rm if}\ a_1=a_2,\\ 0\ &{\rm otherwise}.\end{cases}$
\item Negation coupling: $\mu_\neg(a_1,a_2)=\begin{cases}\frac{1}{2}\ &{\rm if}\ \neg a_1=a_2,\\ 0\ &{\rm otherwise}.\end{cases}$
\end{enumerate}
More generally, let $\mathbf{Unif}_\mathcal{A}$ be the uniform distribution over a finite nonempty set $\mathcal{A}$, i.e. $\mathbf{Unif}_\mathcal{A}(a)=\frac{1}{|\mathcal{A}|}$ for every $a\in\mathcal{A}$. Then each bijection $f:\mathcal{A}\rightarrow\mathcal{A}$ yields a coupling $\mu_f$ for $(\mathbf{Unif}_\mathcal{A},\mathbf{Unif}_\mathcal{A})$:
$$\mu_f(a_1,a_2)=\begin{cases}\frac{1}{|\mathcal{A}|}\ &{\rm if}\ f(a_1)=a_2,\\ 0\ &{\rm otherwise}.\end{cases}$$
\end{Exam}

\begin{Exam}\label{exam-id} For any sub-distribution $\mu$ over $\mathcal{A}$, the identity coupling for $(\mu,\mu)$ is: $\mu_{\rm id}(a_1,a_2)=\begin{cases}\mu(a)\ &{\rm if}\ a_1=a_2 =a,\\ 0\ &{\rm otherwise}.\end{cases}$
\end{Exam}

\begin{Exam}\label{exam-triv} For any distributions $\mu_1,\mu_2$ over $\mathcal{A}_1,\mathcal{A}_2$, respectively, the independent or trivial coupling is:
$\mu_\times (a_1,a_2)=\mu_1(a_1)\cdot\mu_2(a_2).$
\end{Exam}

Obviously, coupling for a pair of distributions is not unique. Then the notion of lifting can be introduced to choose a desirable coupling.

\begin{Def}[Probabilistic Lifting] Let $\mu_1,\mu_2$ be sub-distributions over $\mathcal{A}_1,\mathcal{A}_2$, respectively, and let $\mathcal{R}\subseteq\mathcal{A}_1\times\mathcal{A}_2$ be a relation. Then a sub-distribution $\mu$ over $\mathcal{A}_1\times\mathcal{A}_2$ is called a witness for the $\mathcal{R}$-lifting of $(\mu_1,\mu_2)$ if: \begin{enumerate}\item $\mu$ is a coupling for $(\mu_1,\mu_2)$; \item ${\rm supp}(\mu)\subseteq \mathcal{R}$.\end{enumerate} Whenever a witness exists, we say that $\mu_1$ and $\mu_2$ are related by the $\mathcal{R}$-lifting and write $\mu_1\mathcal{R}^{\#}\mu_2$.
\end{Def}

\begin{Exam} \begin{enumerate}\item Coupling $\mu_f$ in Example \ref{exam-bit} is a witness for the lifting $\mathbf{Unif}_\mathcal{A}$ $\{(a_1,a_2)|f(a_1)=a_2\}^\#\mathbf{Unif}_\mathcal{A}.$
\item Coupling $\mu_{\rm id}$ in Example \ref{exam-id} is a witness for the lifting $\mu =^\# \mu.$
\item Coupling $\mu_\times$ in Example \ref{exam-triv} is a witness for the lifting $\mu_1 T^\# \mu_2$, where $T=\mathcal{A}_1\times\mathcal{A}_2.$\end{enumerate}\end{Exam}

\begin{Prop}\label{pro-basic}\begin{enumerate}\item Let $\mu_1,\mu_2$ be sub-distributions over $\mathcal{A}_1, \mathcal{A}_2$, respectively. If there exists a coupling for $(\mu_1,\mu_2)$, then $|\mu_1|=|\mu_2|.$
\item Let $\mu_1,\mu_2$ be sub-distributions over the same $\mathcal{A}$. Then $\mu_1=\mu_2$ if and only if $\mu_1 =^{\#} \mu_2$.\end{enumerate}
\end{Prop}

\subsection{Quantum Coupling}

With the correspondence of probability distributions/density operators (mixed quantum states) and marginal distributions/partial traces mentioned in the Introduction, we can introduce the notion of quantum coupling. To this end, let us first recall several basic notions from quantum theory; for details, we refer to \cite{NC00}.

Suppose that $\mathcal{H}$ is a finite-dimensional Hilbert space. Let $\mathrm{Herm}(\mathcal{H})$ be the set of Hermitian matrices in $\mathcal{H}$. Let $\mathrm{Pos}(\mathcal{H})$ be the set of positive (semidefinite) matrices in $\mathcal{H}$, and $\mathcal{D}(\mathcal{H})\subset\mathrm{Pos}(\cH)$ is the set of partial density operators, $i.e.$, positive (semidefinite) matrices with trace one.
A positive operator $\rho$ in $\mathcal{H}$ is called a partial density operator if its trace $\mathit{tr}(\rho)=\sum_i\langle i|\rho|i\rangle\leq 1$, where $\{|i\rangle\}$ is an orthonormal basis of $\mathcal{H}$.

We define its support: \begin{align*}\mathrm{supp}(\rho)&=\mathrm{span}\{{\rm eigenvectors\ of}\ \rho\ {\rm with\ nonzero\ eigenvalues}\}\\ &=\mathrm{span}\{|\psi\>\ \big|\ \tr(\rho|\psi\>\<\psi|) = 0\}^\bot.\end{align*}
 If $A$ is an observable, i.e. Hermitian operator, in $\mathcal{H}$, then its expectation in state $\rho$ is $\langle A\rangle_\rho=\mathit{tr}(A\rho).$
Furthermore, let $\mathcal{H}_1, \mathcal{H}_2$ be two Hilbert space. Then partial trace over $\mathcal{H}_1$ is a mapping $\mathit{tr}_1(\cdot)$ from operators in $\mathcal{H}_1\otimes\mathcal{H}_2$ to operators in $\mathcal{H}_2$ defined by
$$\mathit{tr}_1(|\varphi_1\rangle\langle\psi_1|\otimes |\varphi_2\rangle\langle\psi_2|)=\langle\psi_1|\varphi_1\rangle\cdot |\varphi_2\rangle\langle\psi_2|$$ for all $|\varphi_1\rangle,|\psi_1\rangle\in \mathcal{H}_1$ and $|\varphi_2\rangle,|\psi_2\rangle\in\mathcal{H}_2$ together with linearity.
The partial trace $\mathit{tr}_2(\cdot)$ over $\mathcal{H}_2$ can be defined dually.

Now we are ready to define the concept of coupling.

\begin{Def}[Quantum Coupling] Let $\rho_1\in\mathcal{D}(\mathcal{H})$ and $\rho_2\in\mathcal{D}(\mathcal{H}_2)$. Then $\rho\in\mathcal{D}(\mathcal{H}_1\otimes\mathcal{H}_2)$ is called a coupling for $(\rho_1,\rho_2)$ if $\mathit{tr}_1(\rho)=\rho_2$ and $\mathit{tr}_2(\rho)=\rho_1$.\end{Def}

This is actually a very special case of the famous quantum marginal problem, see \cite{Bravyi2004,Chen2014,Carlen2013,YZYY2018} as a very incompleted list for recent development.

\begin{Exam}\label{exam-qunif} Let $\mathcal{H}$ be a Hilbert space and $\mathcal{B}=\{|i\rangle\}$ an orthonormal basis of $\mathcal{H}$. Then the uniform density operator on $\cH$ is
$$\mathbf{Unif}_\mathcal{\cH}=\frac{1}{d}\sum_i|i\rangle\langle i|$$
where $d=\dim\mathcal{H}$ is the dimension of $\mathcal{H}$. Indeed, the uniform density operator on $\cH$ is unique and independent with the choice of orthonormal basis. For each unitary operator $U$ in $\mathcal{H}$, we write $U(\mathcal{B})=\{U|i\rangle\}$, which is also an orthonormal basis of $\mathcal{H}$. Then
$$\rho_U=\frac{1}{d}\sum_i(|i\rangle U|i\rangle)(\langle i|\langle i|U^\dag)$$
is a coupling for $(\mathbf{Unif}_\mathcal{H}, \mathbf{Unif}_{\mathcal{H}})$. In general, for different $U$ and $U^\prime$, $\rho_U \neq \rho_{U^\prime}$, though they are both the couplings for $(\mathbf{Unif}_\mathcal{H}, \mathbf{Unif}_{\mathcal{H}})$.
\end{Exam}

\begin{Exam}\label{exam-qid} Let $\rho$ be a partial density operator in $\mathcal{H}$. Then by the spectral decomposition theorem, $\rho$ can be written as $\rho=\sum_i p_i|i\rangle\langle i|$ for some orthonormal basis $\mathcal{B}=\{|i\rangle\}$ and $p_i\geq 0$ with $\sum_i p_i\leq 1$. We define: $$\rho_{{\rm id}(\mathcal{B})}=\sum_i p_i|ii\rangle\langle ii|.$$ Then it is to see that $\rho_{{\rm id}(\mathcal{B})}$ is a coupling for $(\rho,\rho)$. A difference between this example and Example \ref{exam-id} is that $\rho$ can be decomposed with other orthonormal bases, say $\mathcal{D}=\{|j\rangle\}$: $\rho=\sum_j q_j|j\rangle\langle j|.$ In general, $\rho_{{\rm id}(\mathcal{B})}\neq\rho_{{\rm id}(\mathcal{D})}$, and we can define a different coupling: $$\rho_{{\rm id}(\mathcal{D})}=\sum_jq_j|jj\rangle\langle jj|$$ for $(\rho,\rho)$.
\end{Exam}

\begin{Exam}\label{exam-qtriv} Let $\rho_1\in\mathcal{D}(\mathcal{H}_1)$ and $\rho_2\in\mathcal{D}(\mathcal{H}_2)$ be density operators. Then
tensor product $\rho_\otimes =\rho_1\otimes\rho_2$ is a coupling for $(\rho_1,\rho_2)$.
\end{Exam}

The notion of lifting can also be easily generalised into the quantum setting.

\begin{Def}[Quantum Lifting] Let $\rho_1\in\mathcal{D}(\mathcal{H}_1)$ and $\rho_2\in\mathcal{D}(\mathcal{H}_2)$, and let $\mathcal{X}$ be a subspace of $\mathcal{H}_1\otimes\mathcal{H}_2$. Then $\rho\in\mathcal{D}(\mathcal{H}_1\otimes\mathcal{H}_2)$ is called a witness of the lifting $\rho_1\mathcal{X}^\#\rho$ if: \begin{enumerate}\item $\rho$ is a coupling for $(\rho_1,\rho_2)$; \item $\mathit{supp}(\rho)\subseteq\mathcal{X}$.\end{enumerate}
\end{Def}

\begin{Exam}\begin{enumerate}\item The coupling $\rho_U$ in Example \ref{exam-qunif} is a witness for the lifting: $$\mathbf{Unif}_{\cH}\mathcal{X}(\mathcal{B},U)^\#\mathbf{Unif}_{\cH}$$
where $\mathcal{X}(\mathcal{B},U)=\mathit{span}\{|i\rangle U|i\rangle\}$ is a subspace of $\mathcal{H}\otimes\mathcal{H}.$
\item The coupling $\rho_{{\rm id}(\mathcal{B})}$ in Example \ref{exam-qid} is a witness of the lifting $\rho =_\mathcal{B}^\# \rho$, where $=_\mathcal{B}\ =\mathit{span}\{|ii\rangle\}$ defined by the orthonormal basis $\mathcal{B}=\{|i\rangle\}$ is a subspace of $\mathcal{H}\otimes\mathcal{H}$.
It is interesting to note that the maximal entangled state $|\Psi\rangle=\frac{1}{\sqrt{d}}\sum_i |ii\rangle$ is in $=_\mathcal{B}.$
\item The coupling $\rho_\otimes$ in Example \ref{exam-qtriv} is a witness of the lifting $\rho_1 (\mathcal{H}_1\otimes\mathcal{H}_2)^\# \rho_2$.
\end{enumerate}\end{Exam}

As a quantum generalisation of Proposition \ref{pro-basic}, we have:

\begin{Prop}\begin{enumerate}\item Let $\rho_1\in\mathcal{D}(\mathcal{H}_1)$ and $\rho_2\in\mathcal{D}(\mathcal{H}_2)$. If there exists a coupling for $(\rho_1,\rho_2)$, then $\mathit{tr}(\rho_1)=\mathit{tr}(\rho_2)$.
\item Let $\rho_1,\rho_2\in\mathcal{D}(\mathcal{H})$. Then $\rho_1=\rho_2$ if and only if $\exists$ orthonormal basis $\mathcal{B}$ s.t. $\rho_1 =_\mathcal{B}^\# \rho_2$.
\end{enumerate}\end{Prop}

\begin{proof} Part 1 and Part 2 {($\Rightarrow$)} are obvious. Here, we prove Part 2 {($\Leftarrow$)}. If $\rho_1 =^\#_\mathcal{B} \rho_2$, then there exists a coupling $\rho$ for $(\rho_1,\rho_2)$ such that $\mathit{supp}(\rho)\subseteq\mathit{span}\{|ii\rangle\},$ where $\mathcal{B}=\{|i\rangle\}$. Then we have: $\rho=\sum_j p_j |\Psi_j\rangle\langle\Psi_j|$ for some $|\Psi_j\rangle\in\mathit{span}\{|ii\rangle\}$ and $p_j$. Furthermore, for each $j$, we can write: $|\Psi_j\rangle=\sum_i\alpha_{ji}|ii\rangle.$ Then it is routine to show that $\mathit{tr}_1(|\Psi_j\rangle\langle\Psi_j|)=
\mathit{tr}_2(|\Psi_j\rangle\langle\Psi_j|)=\sum_i |\alpha_{ji}|^2|i\rangle\langle i|.$ Therefore, it holds that $\rho_1=\mathit{tr}_2(\rho)=\sum_j p_j \mathit{tr}_2(|\Psi_j\rangle\langle\Psi_j|)=\sum_j p_j \mathit{tr}_1(|\Psi_j\rangle\langle\Psi_j|)
=\mathit{tr}_1(\rho)=\rho_2.$
\end{proof}

\section{Quantum Strassen Theorem}

As mentioned in the Introduction, a fundamental theorem for probabilistic coupling is the following:

\begin{Thm}[Strassen Theorem]\label{prob-strassen} Let $\mu_1,\mu_2$ be sub-distributions over $\mathcal{A}_1, \mathcal{A}_2$, respectively. Then \begin{equation}\label{strassen}
\mu_1\mathcal{R}^{\#}\mu_2\Rightarrow\forall S\subseteq\mathcal{A}_1.\ \mu_1(S)\leq\mu_2(\mathcal{R}(S))
\end{equation} where $\mathcal{R}(S)$ is the image of $S$ under $\mathcal{R}$:
$\mathcal{R}(S)=\{a_2\in\mathcal{A}_2|\exists a_1\in S\ {\rm s.t.}\ (a_1,a_2)\in\mathcal{R}\}.$ The converse of (\ref{strassen}) holds if $|\mu_1|=|\mu_2|.$
\end{Thm}

In this section, we prove a quantum generalisation of the above Strassen Theorem. For this purpose,
for any subspace $\cX$ of $\cH_1\otimes\cH_2$, we use $P_\cX$ and $P_{\cX^\bot}$ to denote the projections on $\cX$ and $\cX^\bot$ (the ortho-complement of $\cX$), respectively. We use $I_1,I_2, I_{12}$ to denote the identity matrix of $\mathcal{H}_1,\mathcal{H}_2,\mathcal{H}_{12}$, respectively. $\<\cdot,\cdot\>$ is employed to denote the inner product of matrices living in the same space,
\begin{equation*}
\<A,B\>=\tr (A^{\dag}B)
\end{equation*}
Then a quantum Strassen theorem can be stated as follows:

\begin{Thm}[Quantum Strassen Theorem]\label{q-strassen} For any two partial density operators $\rho_1$ in $\cH_1$ and  $\rho_2$ in $\cH_2$ with $\tr(\rho_1) = \tr(\rho_2)$, and for any subspace $\cX$ of $\cH_1\otimes\cH_2$, the following three statements are equivalent:
\begin{itemize}
\item[1.] $\rho_1\cX^\#\rho_2$;
\item[2.] For all observables (Hermitian operators) $Y_1$ in $\cH_1$ and $Y_2$ in $\cH_2$ satisfying $P_{\cX^\bot}\ge Y_1\otimes I_2 - I_1\otimes Y_2$, it holds that
\begin{equation}\label{eq-strassen}\tr(\rho_1Y_1)\leq\tr(\rho_2Y_2).\end{equation}
\item[3.] For all positive observables $Y_1$ in $\cH_1$ and $Y_2$ in $\cH_2$ satisfying $P_{\cX^\bot}\ge Y_1\otimes I_2 - I_1\otimes Y_2$, it holds that
$\tr(\rho_1Y_1)\leq\tr(\rho_2Y_2).$
\end{itemize}
\end{Thm}

\begin{proof} $(\mathit{1}\Rightarrow \mathit{2})$ Suppose $\rho$ is a witness of the lifting $\rho_1\cX^\#\rho_2$. Then for all observables (Hermition operators) $Y_1$ in $\cH_1$ and $Y_2$ in $\cH_2$, if $P_{\cX^\bot}\ge Y_1\otimes I_2 - I_1\otimes Y_2$, then we have:
\begin{align}\label{eq-1}
\tr(\rho_1Y_1) &= \tr(\rho (Y_1\otimes I_2)) \\
\label{eq-2}&\leq \tr(\rho (P_{\cX^\bot} + I_1\otimes Y_2))\\
\label{eq-3}&= \tr(\rho (I_1\otimes Y_2)) \\
\label{eq-4}&= \tr(\rho_2 Y_2).
\end{align}
Equalities (\ref{eq-1}) and (\ref{eq-4}) are derived from the condition that $\rho$ is a coupling for $(\rho_1,\rho_2)$; that is, $\tr_2(\rho)=\rho_1$ and $\tr_1(\rho)=\rho_2$, (\ref{eq-2}) is due to the assumption for $Y_1$ and $Y_2$, and (\ref{eq-3}) is trivial as $\mathrm{supp}(\rho)\subseteq \cX$, so $\tr(\rho P_{\cX^\bot}) = 0$.

{\vskip 4pt}

$(\mathit{2} \Rightarrow \mathit{1})$ Let us first define the semidefinite program $(\Phi, A, B)$:
\begin{center}
  \begin{minipage}{2in}
    \centerline{\underline{Primal problem}}\vspace{-7mm}
    \begin{align*}
      \text{maximize:}\quad & \ip{A}{X}\\
      \text{subject to:}\quad & \Phi(X) = B,\\
      & X\in\pos{\cH_1\otimes\cH_2}.
    \end{align*}
  \end{minipage}
  \hspace*{5mm}
  \begin{minipage}{2in}
    \centerline{\underline{Dual problem}}\vspace{-7mm}
    \begin{align*}
      \text{minimize:}\quad & \ip{B}{Y}\\
      \text{subject to:}\quad & \Phi^{\ast}(Y) \geq A,\\
      & Y\in\mathrm{Herm}\left(\cH_1\oplus\cH_2\right).
    \end{align*}
  \end{minipage}
\end{center}
where:
\begin{align*}
&A = P_\cX,\quad B = \left[
\begin{array}{cc}
\rho_1 &  \\
 & \rho_2 \\
\end{array}
\right], \\
&\Phi(X) = \left[
\begin{array}{cc}
\tr_2(X) &  \\
 & \tr_1(X) \\
\end{array}
\right], \\
&\Phi^{\ast}(Y) = \Phi^{\ast}\left[
\begin{array}{cc}
Y_1 & \cdot \\
\cdot & Y_2 \\
\end{array}
\right]
= Y_1\otimes I_2 + I_1\otimes Y_2.\\
\end{align*}
To show that the above problems are actually primal and dual, respectively, we only need to check the following equality:
\begin{align*}
\forall\ M,\ N,\ \ip{\Phi(M)}{N} &= \tr(\tr_2(M)N_1+\tr_1(M)N_2) \\
&=\tr(M(N_1\otimes I_2) + M(I_1\otimes N_2)) \\
&=\ip{M}{\Phi^{\ast}(N)}.
\end{align*}
Moreover, the strong duality holds for this semidefinite program as we can check that the primal feasible set are not empty and there exists a Hermitian operator $Y$ for which $\Phi^{\ast}(Y) > A$:
\begin{align*}
&\text{Primal feasible set\ } \cA = \left\{X\in\pos{\cH_1\otimes\cH_2}: \Phi(X) = B\right\} \ni {\frac{1}{\tr(\rho_1)}\rho_1\otimes\rho_2} \\
&\text{Choose\ }Y = I_1\oplus I_2 \in\mathrm{Herm}\left(\cH_1\oplus\cH_2\right),\ \Phi^{\ast}(Y) = 2I_{12} > P_\cX.
\end{align*}
So, $\max\ip{P_R}{X} = \min\ip{B}{Y} = \min\{\ip{\rho_1}{Y_1}+\ip{\rho_2}{Y_2}\}$.
Now, let us consider the following condition:
\begin{itemize}
\item[] \textbf{(A)}: For all observable (Hermitian operators) $Y_1$ in $\cH_1$ and $Y_2$ in $\cH_2$ satisfy $Y_1\otimes I_2 + I_1\otimes Y_2 \geq P_\cX$, then
$$\ip{B}{Y} = \tr(\rho_1Y_1+\rho_2Y_2) \geq \tr\rho_1.$$
\end{itemize}
If condition (A) holds, then $\min\ip{B}{Y} \geq \tr\rho_1$. Still remember that $\max\ip{P_\cX}{X} \leq \tr X  = \tr\rho_1$. Due to the strong duality, we have $\max\ip{P_\cX}{X} = \tr\rho_1$. So, $X_{max}$ which maximizes $\ip{P_\cX}{X}$ must satisfy $\ip{P_\cX}{X_{max}} = \tr\rho_1 = \tr X_{max}$. Consequently, $\mathrm{supp}{X_{max}}\subseteq \cX$; in other words, $X_{max}$ is a witness of $\rho_1\cX^\#\rho_2$. Therefore, $\text{condition (A)} \Rightarrow\rho_1R^\#\rho_2.$
On the other hand, condition (A) is equivalent to statement $\mathit{2}$ of the theorem. Indeed, this is not difficult to prove as if we replace $Y_1^\prime = I_1-Y_1$ in condition (A), then
\begin{align*}
Y_1 \in \mathrm{Herm}(\cH_1) &\Longleftrightarrow  Y_1^\prime \in \mathrm{Herm}(\cH_1) \\
Y_1\otimes I_2 + I_1\otimes Y_2 \geq P_\cX &\Longleftrightarrow
I_1\otimes I_2 - P_\cX \geq Y_1^\prime\otimes I_2 - I_1\otimes Y_2 \\
&\Longleftrightarrow
P_\cX^\bot \geq Y_1^\prime\otimes I_2 - I_1\otimes Y_2 \\
\tr(\rho_1Y_1+\rho_2Y_2) \geq \tr\rho_1 &\Longleftrightarrow
\tr(\rho_2Y_2) \geq \tr(\rho_1I_1) - \tr(\rho_1(I_1-Y_1^\prime)) \\
&\Longleftrightarrow
\tr(\rho_2Y_2) \geq \tr(\rho_1Y_1^\prime).
\end{align*}
From the above, we can directly derive statement $\mathit{2}$. In summary, we have:
$\text{statement }\mathit{2} \Leftrightarrow \text{condition (A)} \Rightarrow\rho_1R^\#\rho_2.$

{\vskip 4pt}

$(\mathit{2} \Rightarrow \mathit{3})$ Obvious.

{\vskip 4pt}

$(\mathit{3} \Rightarrow \mathit{2})$ We only need to show that, for any two observables $Y_1$ in $\cH_1$ and $Y_2$ in $\cH_2$ satisfy $P_\cX^\bot\ge Y_1\otimes I_2 - I_1\otimes Y_2$, there exist two positive observables $Y_1^\prime$ in $\cH_1$ and $Y_2^\prime$ in $\cH_2$ such that $P_\cX^\bot\ge Y_1^\prime\otimes I_2 - I_1\otimes Y_2^\prime$ and
$$\tr(\rho_1Y_1)\leq\tr(\rho_2Y_2)  \Longleftrightarrow \tr(\rho_1Y_1^\prime)\leq\tr(\rho_2Y_2^\prime).$$
Note that $Y_1$ and $Y_2$ are Hermitian, so their eigenvalues are real, and we can define $\lambda = \min\{\mathrm{eigenvalues\ of\ }Y_1 \text{\ and\ } Y_2\}$. Choose $Y_1^\prime = Y_1 - \lambda I_1$ and $Y_2^\prime = Y_2 - \lambda I_2$. Obviously,  $Y_1^\prime$ and $Y_2^\prime$ are positive observables, and satisfy
\begin{align*}
P_\cX^\bot &\ge Y_1\otimes I_2 - I_1\otimes Y_2 \\
&= Y_1\otimes I_2 - \lambda I_1 \otimes I_2 + \lambda I_1 \otimes I_2 - I_1\otimes Y_2 \\
&= Y_1^\prime\otimes I_2 - I_1\otimes Y_2^\prime.
\end{align*}
Moreover, as $\tr(\rho_1) = \tr(\rho_2)$, we have
\begin{align*}
\tr(\rho_1Y_1)\leq\tr(\rho_2Y_2) &\Longleftrightarrow
\tr(\rho_1Y_1) - \tr(\rho_1\lambda I_1)\leq\tr(\rho_2Y_2) - \tr(\rho_2\lambda I_2) \\
&\Longleftrightarrow
\tr(\rho_1Y_1^\prime)\leq\tr(\rho_2Y_2^\prime).
\end{align*}
\end{proof}

{\textbf{Remark:}In the above proof, it is indeed naturally to employing the methods of semidefinite programming. In \cite{Hsu17}, Hsu deliberately constructs a flow network, and then using the max-flow min-cut theorem to prove the Strassen theorem in the finite case.  Essentially, the max-flow min-cut theorem is a special case of the duality theorem for linear programs (LP). Considering the fact that quantum states, quantum operations and so on are all described by matrices, similar to LP, semi-definite programming (SDP) is a powerful and widely used method of convex optimization in quantum theory. Indeed, when all matrices appeared in a SDP are diagonal, then the SDP reduces to LP. In the following section, we will see that in the degenerate case, quantum Strassen theorem also reduces to the classical Strassen theorem.}

\section{Classical Reduction of Quantum Strassen Theorem}

At the first glance, Theorem \ref{prob-strassen} (Strassen Theorem for Probabilistic Coupling) and Theorem \ref{q-strassen} (Quantum Strassen Theorem) are very different. In this section, we show that Theorem \ref{q-strassen} is indeed a quantum generalisation of Theorem \ref{prob-strassen}. 

To this end, let $\mu_1$ be a {sub-}distribution over $[m]$ ($[m] = \{i\in \mathbb{N}\ \big| \ 1\le i\le m\}$) and $\mu_2$ over $[n]$. {And the corresponding degenerate partial density operators (quantum states) are}:
\begin{align*}
\rho_1 = \left[
\begin{array}{cccc}
\mu_1(1) & & & \\
& \mu_1(2) & & \\
& & \ddots & \\
& & & \mu_1(m) \\
\end{array}
\right], \quad
\rho_2 = \left[
\begin{array}{cccc}
\mu_2(1) & & & \\
& \mu_2(2) & & \\
& & \ddots & \\
& & & \mu_2(n) \\
\end{array}
\right]
\end{align*} in $\cH_1=\mathrm{span}\{|i\rangle: i\in [m]\}$ and $\cH_2=\mathrm{span}\{|j\rangle:j\in [n]\}$, respectively.
Furthermore, let $R\subseteq\{(i,j) \big| i\in[m],\ j\in[n]\}$ be a classical relation from $[m]$ to $[n]$. Then the corresponding (quantum relation) subspace of $\cH_1\otimes\cH_2$ is defined as
$$\cX_R = \mathrm{span}\{|i\>|j\>\ \big|\ (i,j)\in R\}.$$

{Based on the above definition of the degenerate case, in the rest part of this section, Proposition \ref{Prop1} shows that the left hand side of Eqn.(\ref{strassen}) in Theorem \ref{prob-strassen} is equivalent to the statement $\mathit{1}$ in Theorem \ref{q-strassen}, while Proposition \ref{Prop2} states the equivalence of the right hand side of Eqn.(\ref{strassen}) in Theorem \ref{prob-strassen} and the statement $\mathit{3}$ in Theorem \ref{q-strassen}, concluding that Theorem \ref{prob-strassen} (Strassen Theorem) is indeed a reduction of Theorem \ref{q-strassen} (Quantum Strassen Theorem).}

The following proposition indicates that probabilistic lifting is a special case of quantum lifting.
\begin{Prop}
\label{Prop1}
 $\mu_1R^\#\mu_2 \Longleftrightarrow \rho_1\cX_R^\#\rho_2.$
\end{Prop}
\begin{proof} ($\Rightarrow$) Suppose that there is a witness $\mu$ of the lifting $\mu_1R^\#\mu_2$. We define
the {partial} density operator:
$$
\rho:\<i|\<j|\rho|i^\prime\>|j^\prime\> = \left\{
\begin{array}{ll}
\mu(i,j) & i=i^\prime,\ j=j^\prime\\
0 & i\neq i^\prime \text{\ or\ } j\neq j^\prime
\end{array}
\right..
$$
It is easy to check:
\begin{align*}
&\<i|\tr_2(\rho)|i^\prime\> = \sum_{j = 1}^{n}\<i|\<j|\rho|i^\prime\>|j\> = \left\{
\begin{array}{ll}
\sum_{j = 1}^{n}\mu(i,j)=\mu_1(i) & i=i^\prime\\
0 & i\neq i^\prime
\end{array}
\right., \\
&\<j|\tr_2(\rho)|j^\prime\> = \sum_{i = 1}^{n}\<i|\<j|\rho|i\>|j^\prime\> = \left\{
\begin{array}{ll}
\sum_{i = 1}^{m}\mu(i,j)=\mu_2(j) & j=j^\prime\\
0 & j\neq j^\prime
\end{array}
\right.. \\
\end{align*}
So, $\tr_2(\rho) = \rho_1$ and $\tr_1(\rho) = \rho_2$; that is, $\rho$ is a coupling for $(\rho_1,\rho_2)$. Furthermore, we have:
\begin{align*}
\tr(\rho P_{\cX_R}) &= \sum_{(i,j)\in R}\<i|\<j|\rho|i\>|j\> \\
&= \sum_{(i,j)\in R}\mu(i,j) \\
&= \sum_{(i,j)\in R}\mu(i,j) + \sum_{(i,j)\notin R}\mu(i,j) \\
&= \tr(\rho) \\
\end{align*}
Thus, $\mathrm{supp}(\rho)\subseteq \cX_R$, and $\rho$ is a witness of the quantum lifting $\rho_1\cX_R^\#\rho_2$.

{\vskip 4pt}

($\Leftarrow$) Suppose there is a witness $\rho$ of the quantum lifting $\rho_1\cX_R^\#\rho_2$. Let us construct
the joint {sub-}distribution $\mu$:
$$
\mu(i,j) = \<i|\<j|\rho|i\>|j\>\ {\rm for\ all}\ i,j.
$$
It is easy to check:
\begin{align*}
&\sum_{j = 1}^{n}\mu(i,j) = \sum_{j = 1}^{n}\<i|\<j|\rho|i\>|j\> = \<i|\rho_1|i\> = \mu_1(i), \\
&\sum_{i = 1}^{m}\mu(i,j) = \sum_{i = 1}^{m}\<i|\<j|\rho|i\>|j\> = \<j|\rho_2|j\> = \mu_2(j).
\end{align*}
Also, if $(i,j)\notin R$, then $|i\>|j\>\perp \cX_R$, then
$$
\mu(i,j) = \<i|\<j|\rho|i\>|j\> = \tr(\rho|i\>|j\>\<i|\<j|) = 0
$$
as $\mathrm{supp}(\rho)\subseteq \cX_R$. Thus, $\mathrm{supp}(\mu)\subseteq R$, and $\mu$ is a witness of the lifting $\mu_1R^\#\mu_2$.
\end{proof}

The following proposition further shows that in the degenerate case,inequality (\ref{eq-strassen}) to (\ref{strassen}).
Surprisingly, such a reduction can be realized even without the condition of lifting.
\begin{Prop}
\label{Prop2} Two statements are equivalent:
\begin{itemize}
\item[1.] For any $S\subseteq[m]$, $\mu_1(S)\leq\mu_2(R(S))$;
\item[2.] For all positive observables $Y_1$ in $\cH_1$ and $Y_2$ in $\cH_2$ satisfy $P_{\cX_R}^\bot\ge Y_1\otimes I_2 - I_1\otimes Y_2$, then
$$\tr(\rho_1Y_1)\leq\tr(\rho_2Y_2)$$
\end{itemize}
\end{Prop}
\begin{proof}
As $\rho_1$, $\rho_2$ and {$P_{\cX_R}$} are diagonal density operators, so we only need to consider those $Y_1$ and $Y_2$ which are also diagonal.
We use the notation $Y_{1,i} = (Y_1)_{ii}$ and $Y_{2,j} = (Y_2)_{jj}$ for simplicity. Then it holds that
\begin{align*}
P_{\cX_R}^\bot\ge Y_1\otimes I_2 - I_1\otimes Y_2
\Longleftrightarrow \forall\ i,j\ \left\{
\begin{array}{ll}
Y_{2,j} \ge Y_{1,i} & (i,j)\in R \\
Y_{2,j} \ge Y_{1,i} - 1 & (i,j)\notin R \\
\end{array}
\right. \\
\end{align*}
Now we need a technical lemma:
 \begin{Lem}
 \label{Lemma1}
 The following two statements are equivalent:
 \begin{itemize}
 \item[$1^\prime$.]
 If $Z_{1,i}\in\{0,1\}$, $Z_{2,j}\in\{0,1\}$, then
 $$
 \forall\ i,j\ \left\{
\begin{array}{ll}
Z_{2,j} \ge Z_{1,i} & (i,j)\in R \\
Z_{2,j} \ge Z_{1,i} - 1 & (i,j)\notin R \\
\end{array}
\right. \ \Rightarrow\
\sum_{i=1}^{m}\mu_1(i)Z_{1,i}\le\sum_{j=1}^{n}\mu_2(j)Z_{2,j}
 $$
 \item[$2^\prime$.]
 If $Y_{1,i}\ge0$, $Y_{2,j}\ge0$, then
 $$
 \forall\ i,j\ \left\{
\begin{array}{ll}
Y_{2,j} \ge Y_{1,i} & (i,j)\in R \\
Y_{2,j} \ge Y_{1,i} - 1 & (i,j)\notin R \\
\end{array}
\right. \ \Rightarrow\
\sum_{i=1}^{m}\mu_1(i)Y_{1,i}\le\sum_{j=1}^{n}\mu_2(j)Y_{2,j}
 $$
 \end{itemize}
{where $Z_1, Z_2$ are also diagonal matrices, and $Z_{1,i} = (Z_1)_{ii}$, $Z_{2,j} = (Z_2)_{jj}$.}
 \end{Lem}
For readability, let us first use this lemma to finish the proof of the proposition, but postpone the proof of the lemma itself to the end of this section. As
\begin{align*}
\tr(\rho_1Y_1) &= \sum_{i=1}^{m}(\rho_1)_{ii}(Y_1)_{ii} = \sum_{i=1}^{m}\mu_1(i)Y_{1,i}, \\
\tr(\rho_2Y_2) &= \sum_{j=1}^{m}(\rho_2)_{jj}(Y_2)_{jj} = \sum_{j=1}^{n}\mu_2(j)Y_{2,j}, \\
\end{align*} it is direct to see that statement 2 of the proposition is equivalent to statement 2$^\prime$ of the above lemma.
For the statement 1$^\prime$ of the above, we can define the set $S = \{i \in [m]\ \big|\ Z_{1,i} = 1\}$ and
$T = \{j \in [n]\ \big|\ Z_{2,j} = 1\}$, then it is equivalent to:
$$
\forall\ S\subseteq [m],\ T\subseteq [n],\ R(S)\subseteq T\ \Rightarrow\ \mu_1(S)\le\mu_2(T),
$$
which is exactly the statement 1 of the proposition.
Therefore, using the above lemma, we see that statements 1 and 2 in the proposition are equivalent.
\end{proof}

Combining Propositions \ref{Prop1} and \ref{Prop2}, we see that Theorem \ref{prob-strassen} (Strassen Theorem for probabilistic coupling) is a reduction of Theorem \ref{q-strassen} (Quantum Strassen Theorem).

To conclude this section, let us present the following:

\begin{proof}[Proof of Lemma \ref{Lemma1}] (2$^\prime$ $\Rightarrow$ 1$^\prime$) This is trivial as statement 1$^\prime$ is a special case of statement 2$^\prime$.

{\vskip 4pt}

(1$^\prime$ $\Rightarrow$ 2$^\prime$) For any $Y_1$, we can construct a decreasing sequence $Z_{11}>\cdots > Z_{1k} > Z_{1(k+1)}>\cdots$ such that:
$$Y_1 = \sum_{k}\lambda_k Z_{1k}, \quad \lambda_k\ge0.$$
We further define $S_k = \{i \in [m]\ \big|\ Z_{1k,i} = 1\}$ and the corresponding $T_k = R(S_k)$.
Then, $\{S_k\}$ is a strictly decreasing sequence; that is, $S_k\supset S_{k+1}$ for all $k$, and $\{T_k\}$ is a non-increasing sequence; that is, $T_k\supseteq T_{k+1}$ for all $k$.
Let us also define $Z_{2k}$:
$$Z_{2k,j} = \left\{
\begin{array}{ll}
1 & j\in T_k \\
0 & j\notin T_k \\
\end{array}
\right.
$$
and a new operator $Y_{2min}$:
$$Y_{2min} = \sum_k\lambda_kZ_{2k}.$$
Note that any pair of $Z_{1k}$ and $Z_{2k}$ satisfy statement $\mathit{1^\prime}$. Then
\begin{align*}
\sum_{i=1}^{m}\mu_1(i)Y_{1,i} &= \sum_{i=1}^{m}\mu_1(i)\sum_k\lambda_k Z_{1k,i} =\sum_k\lambda_k\sum_{i=1}^{m}\mu_1(i)Z_{1k,i} \\
&\leq \sum_k\lambda_k\sum_{j=1}^{n}\mu_2(j)Z_{2k,j}
=\sum_{j=1}^{n}\mu_2(j)\sum_k\lambda_kZ_{2k,j} \\
&=\sum_{j=1}^{n}\mu_2(j)Y_{2min,j}
\end{align*}
Now it suffices to prove that for any $Y_2$ satisfying the condition in statement $\mathit{2^\prime}$, we have $Y_2 \ge Y_{2min}$. To show this, let us use $\mathbb{I}(\cdot)$ to represent the indication function, and consider the following two cases:
\begin{itemize}
\item Case 1: $t \notin T_1$. Then of course, $\forall\ k:\ t\notin T_k$, so,
$$Y_{2min,t} = \sum_k\lambda_kZ_{2k,t} = \sum_k\lambda_k\mathbb{I}(t\in T_k) = 0 \le Y_{2,t}$$ \\
\item Case 2: $\exists k$ such that  $t\in T_k$. Suppose $k_t = \max\{k: t\in T_k\}$. Then we have the following two facts: (1) for $k\le k_t$, $t\in T_k$ and for $k>k_t$. $t\notin T_k$; (2) $\exists s\in S_{k_t}$ such that $(s,t)\in R$, and for $k\le k_t$, $s\in S_k$. Combining these two facts, wel have:\\
    \begin{align*}
    Y_{2,t} &\ge Y_{1,s} = \sum_k\lambda_kZ_{1k,s} = \sum_k\lambda_k\mathbb{I}(s\in S_k) \\
    &\ge \sum_{k = 1}^{k_t}\lambda_k\mathbb{I}(s\in S_k) = \sum_{k = 1}^{k_t}\lambda_k\mathbb{I}(t\in T_k) \\
    &= \sum_{k}\lambda_k\mathbb{I}(t\in T_k) = \sum_k\lambda_kZ_{2k,t}\\
    &= Y_{2min,t} \\
    \end{align*}
\end{itemize}
So, for any $Y_2$ satisfies the condition in statement $\mathit{2^\prime}$, we have:
\begin{align*}
\sum_{j=1}^{n}\mu_2(j)Y_{2,j} \ge \sum_{j=1}^{n}\mu_2(j)Y_{2min,j} \ge \sum_{i=1}^{m}\mu_1(i)Y_{1,i}.
\end{align*}
\end{proof}

\section{Conclusion}

In this paper, we defined the notion of quantum coupling and proved a quantum generalisation of Strassen theorem for probabilistic coupling. It is well-known that Strassen theorem is true in both the finite-dimensional and infinite-dimensional cases. However, Theorem \ref{q-strassen} (quantum Strassen theorem) was proved only in the finite-dimensional case. So, an open problem is: whether is quantum Strassen theorem still valid in the infinite-dimensional case?
Another interesting topic for further study is to use the coupling techniques to study the behaviours of quantum random walks and quantum Markov chains. As pointed out in the Introduction, in the future studies, we hope
to apply quantum coupling to develop quantum relational Hoare logic and then use it in formal verification of quantum cryptographic protocols and differential privacy.

\section{Acknowledgment}

This work was partly supported by the Australian Research Council (Grant No: DP160101652, DE180100156) and the Key Research Program of Frontier Sciences, Chinese Academy of Sciences (Grant No: QYZDJ-SSW-SYS003).

\end{document}